\documentclass[conference,letterpaper]{IEEEtran}

\usepackage{amsmath,amssymb,amsfonts,mathtools}
\usepackage{cite}
\usepackage{graphicx}
\usepackage{bm}
\usepackage{booktabs}
\usepackage{array}
\usepackage{algorithm}
\usepackage{algpseudocode}
\makeatletter
\newcommand\fs@shiftedruled{\def\@fs@cfont{\bfseries}\let\@fs@capt\floatc@ruled
  \def\@fs@pre{\kern6pt\hrule height.8pt depth0pt \kern2pt}%
  \def\@fs@post{\kern2pt\hrule\relax}%
  \def\@fs@mid{\kern2pt\hrule\kern2pt}%
  \let\@fs@iftopcapt\iftrue}
\makeatother
\floatstyle{shiftedruled}
\restylefloat{algorithm}

\newcommand{\R}{\mathbb{R}}
\newcommand{\Z}{\mathbb{Z}}
\newcommand{\T}{\mathbb{T}}
\newcommand{\E}{\mathbb{E}}
\newcommand{\Ann}{\operatorname{Ann}}
\newcommand{\BLER}{\mathrm{BLER}}
\newtheorem{proposition}{Proposition}
\newtheorem{corollary}{Corollary}

\title{Key-Reuse Vulnerability of Phase-Keyed Fourier-Curve Modulation: Relation Leakage and Key-Refresh Cost on Coded Links}

\author{
\IEEEauthorblockN{Bin~Han\IEEEauthorrefmark{1},
Muxia~Sun\IEEEauthorrefmark{2},
and Hans~D.~Schotten\IEEEauthorrefmark{1}\IEEEauthorrefmark{4}}
\IEEEauthorblockA{\IEEEauthorrefmark{1}RPTU University Kaiserslautern-Landau, Germany,\quad
\IEEEauthorrefmark{2}Beijing Huairou Laboratory, China\\
\IEEEauthorrefmark{4}German Research Center for Artificial Intelligence (DFKI), Germany}
}

\begin{document}
\maketitle

\begin{abstract}
Keyed modulation is often judged secure by its key-space size and a key-less receiver's error rate, evidence that fails when the key is reused and the waveform is harmonically coupled. For a phase-keyed Fourier-curve constellation, whose $k$ tones share one data parameter, integer relations among harmonic indices yield data-cancelling mixed moments that expose key characters. A modular relation lattice characterizes them; for consecutive harmonics, third-order moments recover the relative phases and a fourth-order moment completes the key up to cyclic relabeling whenever its coefficient is nonzero, as in all evaluated settings. A non-data-aided relation-moment estimator with code-aided resolution of the relabeling turns this leakage into an attack without key-space enumeration. With one key per 168-symbol codeword, the attack meets a predeclared $0.1$ compromise criterion at eleven of twelve operating points of a $(3,6)$ LDPC-coded link; at the middle noise level, the eavesdropper's block error rate is below $0.04$, and $0.03$ for consecutive harmonics without artificial noise when a key spans only $28$ symbols. Tangent artificial noise and a harmonic set without relations below order four raise her error rate at intermediate reuse lengths but do not remove the one-codeword vulnerability. On a $2^{128}$-key grid at the same noise level, the cheapest tested refresh schedule whose $95\%$ lower confidence bound on her block error rate reaches $0.9$ costs $3.54$ fresh key bits per information bit, and no tested schedule on either grid meets even a $0.5$ target with less key than a one-time pad.
\end{abstract}


\section{Introduction}\label{sec:intro}

Keyed modulation and physical-layer encryption let a shared secret select the rotation, mapping, or shape of the transmitted constellation, so that only a receiver holding the key demodulates correctly~\cite{zhang2017ple,mukherjee2014,sultan2018}. The security of such schemes is often supported by two observations: the key space is large, and a receiver that uses a wrong key suffers a high error rate. Both are insufficient security evidence when a key is reused and the waveform admits low-order statistics that cancel the data. As a simple instance, let three received tones carry the phases $h_r\theta+\varphi_r$, $h_s\theta+\varphi_s$ and $h_{r+s}\theta+\varphi_{r+s}$ of one data parameter $\theta$ with $h_r+h_s=h_{r+s}$. The data phase then cancels from the product $Z_rZ_s\overline Z_{r+s}$, whose conditional mean is a known real coefficient times $e^{\mathrm i(\varphi_r+\varphi_s-\varphi_{r+s})}$ for every transmitted symbol, so averaging the product over the symbols that share the key estimates this function of the key without any search. Fig.~\ref{fig:relation} shows this for $k=2$ tones with $h=(1,2)$.

Data-cancelling statistics are classical tools. Power-law carrier recovery removes the modulation of symmetric constellations~\cite{viterbi1983,moeneclaey1994}, cumulants of the received signal classify modulation formats~\cite{swami2000}, bispectral methods reconstruct phases~\cite{nikias1993,bendory2018}, and code-aided estimators exploit the channel code to acquire synchronization parameters at low SNR~\cite{herzet2007}. Keyed designs differ in whether such statistics apply. A common secret rotation of a symmetric constellation is recovered up to its symmetry by the power-law estimator, whereas a symbolwise remapping driven by a cryptographic generator~\cite{zhang2011chaos,sultan2018} inherits its security from that generator. Fresh keys can also be generated from reciprocal channels~\cite{maurer1993,zhang2016keygen}, at the net secret-key rate that remains after probing, reconciliation, and privacy amplification. These moment tools are thus standard; the contribution lies in identifying which key characters a harmonically coupled waveform exposes and in carrying this diagnosis through a finite-key coded link to its entropy cost.

We study this effect for phase-keyed Fourier-curve modulation~\cite{han2026letter}. A symbol occupies $k$ complex tones whose phases are integer multiples $h_j\theta$ of one data parameter $\theta$, offset by a secret phase vector $\bm\varphi$; tangent artificial noise (AN)~\cite{goel2008} can be added along the curve. The companion letter~\cite{han2026letter} analyzed detection when the receiver knows the key, and our preprint~\cite{han2026arxiv} developed the coded receiver. The present paper treats the complementary problem: how an eavesdropper who does not know the key exploits its reuse, and how much independent secret entropy is consumed by refresh schedules that frustrate the evaluated attack.

\begin{figure}[!t]
\centering
\vspace*{3pt}
\includegraphics[width=\linewidth]{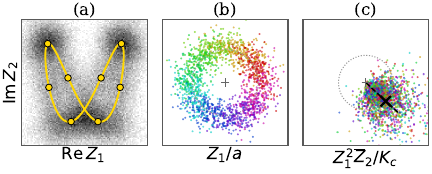}
\caption{Relation leakage at $k=2$, $h=(1,2)$, $M=8$, $\beta=0.3$, $\sigma_c=0.14$. (a) Received densities of the keyed curve with tangent AN, projected to $(\mathrm{Re}\,Z_1,\mathrm{Im}\,Z_2)$, in the geometry of~\cite{han2026arxiv}. (b) Tone~1, colored by the transmitted symbol, spreads over the circle with the data. (c) The monomial $Z_1^2\overline Z_2$ concentrates at the key phase $2\varphi_1-\varphi_2$ (cross) for every symbol.}
\label{fig:relation}
\end{figure}

The contributions are threefold. First, a modular relation lattice of the harmonic indices characterizes which key characters are exposed by data-cancelling mixed moments and which ambiguity remains; for consecutive harmonics, third-order relations recover the relative phases and a fourth-order relation completes the key up to the cyclic relabeling of the codebook whenever its coefficient is nonzero, and a comparison family with odd harmonics has no exposing relation below order four. Second, a non-data-aided relation-moment estimator, followed by a code-aided resolution of the residual relabeling, turns the leakage into an attack, which we evaluate by the eavesdropper's coded block error rate on an LDPC-coded link as a function of the number of symbols that share one key. Third, for an explicit finite product-grid protocol we relate periodic independent refresh to the secret entropy it consumes and compare it with a one-time pad on the information bits~\cite{shannon1949}. Tangent AN and the odd harmonic family enter as robustness checks of this chain.

The relation argument applies to waveforms whose observed components share one latent data parameter, and Table~\ref{tab:scope} delimits it from related keyed designs.

The remainder of the paper is organized as follows. Sec.~\ref{sec:model} states the signal, protocol, and threat model; Sec.~\ref{sec:leak} derives the relation leakage; Sec.~\ref{sec:attack} describes the attack; Sec.~\ref{sec:results} reports the coded experiments and the key-refresh cost; Sec.~\ref{sec:disc} discusses implications and limitations, and Sec.~\ref{sec:concl} concludes.

\section{System, Protocol, and Threat Model}\label{sec:model}

\subsection{Phase-keyed harmonic modulation}

Let $h=(h_1,\ldots,h_k)$ be distinct positive integers, $\theta_s=2\pi s/M$ for $s\in\Z_M$, and $\bm\varphi\in\T^k$ the phase key, $\T=\R/2\pi\Z$. With AN fraction $\beta\in[0,1)$, $v_h^2=\frac1k\sum_j h_j^2$, $a=\sqrt{(1-\beta)/k}$ and $d=\sqrt\beta/(\sqrt k\,v_h)$, the $j$-th received tone of symbol $S_\ell=s$ is
\begin{equation}
Z_{\ell,j}=e^{\mathrm i(\varphi_j+h_j\theta_s)}\bigl(a+\mathrm i h_j d A_\ell\bigr)+W_{\ell,j},
\label{eq:tone}
\end{equation}
where $A_\ell\sim\mathcal N(0,1)$ is the scalar tangent AN sample of the transmitted waveform, drawn independently for every symbol, and $W_{\ell,j}$ is proper complex Gaussian receiver noise with $\E|W_{\ell,j}|^2=2\sigma_c^2$, independent across tones and symbols and independent of the message, the key, and $A_\ell$. The consecutive family $h=(1,\ldots,k)$ is the Fourier curve of~\cite{han2026letter}; the odd family $h=(1,3,\ldots,2k-1)$ serves as a comparison. Bob knows the key and uses the covariance-aware max-log demapper of~\cite{han2026arxiv}, built on the matched metric of~\cite{han2026letter}; we use it only as the legitimate receiver and as the decoder Eve applies after estimating the key.

\subsection{Key reuse and finite key protocol}

One LDPC codeword occupies $n_s$ symbols, split into $n_s/L$ segments of $L$ symbols; each segment uses an independent key. The signal set is invariant under $\bm\varphi\mapsto\bm\varphi-2\pi q h/M$, $q\in\Z_M$, which only relabels the symbols. With $h_1=1$, the protocol draws every key in the fundamental domain $\mathcal F$ of keys with $\varphi_1\in[0,2\pi/M)$ and fixes the labeling publicly, so no secret label shift is added; continuous keys are drawn uniformly on $\mathcal F$. For the key-cost analysis the key is drawn uniformly from the product grid
\begin{equation}
\mathcal K_b=\Bigl\{\bigl(\tfrac{2\pi q_1}{M2^b},\tfrac{2\pi q_2}{2^b},\ldots,\tfrac{2\pi q_k}{2^b}\bigr):q_j\in\{0,\ldots,2^b-1\}\Bigr\}
\label{eq:grid}
\end{equation}
with entropy $H(K_b)=bk$ bits per refresh; $b=8$ at $k=16$ gives a product grid of $2^{128}$ protocol keys. This grid serves as the test protocol; its size counts protocol keys, and the attack of Sec.~\ref{sec:attack} needs no search over them.

The canonical labeling is part of the security model: since $h_1=1$ and $\varphi_1$ stays within one sector of width $2\pi/M$, the phase of tone~1 alone places the symbol in its sector. With a fresh key per symbol, tones $2,\ldots,k$ carry no information, so Eve's information per symbol is $I(S;Z_1)$, which by quadrature lies between $0.88$ and $1.85$ bits per symbol at the tested noise levels ($1.44$ at $\sigma_c=0.13$ without AN), below the $3.01$ bits per symbol of the code; it reaches the code rate at $11.7$ to $13.8$\,dB per complex slot. A full-circle first coordinate, a secret label shift of $\log_2M$ further key bits per refresh, closes this channel and leaves Eve to resolve the label among $M$ candidates with the code.

\subsection{Threat model and metric}

Eve knows the harmonic set, $M$, $\beta$, the noise level, the labeling, the LDPC code and decoder, Bob's demapper, the segment boundaries, and the attack; she does not know the segment keys, which stay constant within a segment, and receives no pilot, preamble, header, CRC, or known payload under a key. The AN sample $A_\ell$ is common to Bob and Eve because it is part of the transmitted waveform. The primary metric is Eve's coded block error rate $\BLER_E$; key-estimation error serves only as a diagnostic. A low $\BLER_E$ demonstrates compromise by the evaluated attack, and a high $\BLER_E$ measures resistance to this attack; a decoder that fails on a block can still disclose much of its content~\cite{duffy2023}.

\section{Relation Leakage}\label{sec:leak}

For $c\in\Z^k$ write $c=c^+-c^-$ with $c^\pm_j=\max(\pm c_j,0)$ and define the signed mixed monomial $U_{c,\ell}=\prod_j Z_{\ell,j}^{c_j^+}\,\overline{Z}_{\ell,j}^{\,c_j^-}$ and the modular relation lattice $L_M(h)=\{c\in\Z^k:\langle c,h\rangle\in M\Z\}$.

\begin{proposition}\label{prop:char}
For every $c\in\Z^k$,
\begin{equation}
\E\bigl[U_{c,\ell}\mid S_\ell=s,\bm\varphi\bigr]=K_c\,e^{\mathrm i\langle c,\bm\varphi\rangle}e^{\mathrm i\theta_s\langle c,h\rangle},
\label{eq:char}
\end{equation}
with the real coefficient $K_c=\E_A\bigl[\prod_j(a+\mathrm ih_jdA)^{c_j^+}(a-\mathrm ih_jdA)^{c_j^-}\bigr]$. If $c\in L_M(h)$, the conditional law of $U_{c,\ell}$ does not depend on $S_\ell$, so the data cancel. If in addition $K_c\neq0$, the mean reveals the key character $e^{\mathrm i\langle c,\bm\varphi\rangle}$.
\end{proposition}
\begin{IEEEproof}
Let $\widetilde W_{\ell,j}=e^{-\mathrm i(\varphi_j+h_j\theta_s)}W_{\ell,j}$. Proper circular noise is invariant under rotation, so $(\widetilde W_{\ell,1},\ldots,\widetilde W_{\ell,k})$ has the law of $(W_{\ell,1},\ldots,W_{\ell,k})$, and
\begin{multline}
U_{c,\ell}=e^{\mathrm i\langle c,\bm\varphi\rangle}e^{\mathrm i\theta_s\langle c,h\rangle}\\
\times\prod_j\bigl(a+\mathrm ih_jdA_\ell+\widetilde W_{\ell,j}\bigr)^{c_j^+}\bigl(a-\mathrm ih_jdA_\ell+\overline{\widetilde W}_{\ell,j}\bigr)^{c_j^-}.
\label{eq:rot}
\end{multline}
Since $c_j^+c_j^-=0$, each tone contributes only pure powers of its rotated noise or of the conjugate, whose expectations vanish, and independence across tones removes every term that contains noise. The average over $A$ is $K_c$, which is real because $A\mapsto-A$ conjugates the integrand. For $c\in L_M(h)$, $\theta_s\langle c,h\rangle\in2\pi\Z$ for all $s$, so by~\eqref{eq:rot} the whole law of $U_{c,\ell}$ is independent of $s$.
\end{IEEEproof}
Since $A_\ell$ and the receiver noise are independent across symbols, the monomials of one segment are independent and identically distributed given the key, even when the symbols form a codeword. Proposition~\ref{prop:char} characterizes the key characters exposed by signed monomial moments.

\begin{proposition}\label{prop:ann}
Let $R\subseteq L_M(h)$ be the integer span of the relations whose characters are observed. Two keys produce the same characters if and only if they differ by an element of $\Ann(R)=\{\bm\omega\in\T^k:e^{\mathrm i\langle c,\bm\omega\rangle}=1\ \forall c\in R\}$. With $g=\gcd(M,h_1,\ldots,h_k)$ and $R=L_M(h)$, $\Ann(R)$ is the cyclic relabeling group $\{-2\pi qh/M:q=0,\ldots,M/g-1\}$ of order $M/g$.
\end{proposition}
\begin{IEEEproof}
The first claim restates equality of all characters. The map $c\mapsto\langle c,h\rangle\bmod M$ has kernel $L_M(h)$ and image $g\Z_M$ of order $M/g$, so $[\Z^k:L_M(h)]=M/g$; by character duality the annihilator of the kernel is the group generated by $-2\pi h/M$.
\end{IEEEproof}
Both evaluated families have $h_1=1$, so $g=1$ and the ambiguity has order $M$. Equal real rank is not enough: $R=L_M(h)$ holds if and only if $\Ann(R)=\Ann(L_M(h))$, and a proper sublattice of finite index leaves an additional discrete ambiguity, so generation of $L_M(h)$ must be checked over the integers.

\begin{corollary}[consecutive harmonics]\label{cor:cons}
For $h=(1,\ldots,k)$ the relations $e_1+e_{m-1}-e_m$, $m=2,\ldots,k$, generate $\{c:\langle c,h\rangle=0\}$, and their third-order moments $\E[Z_rZ_s\overline Z_{r+s}]=a[a^2+d^2(r^2+rs+s^2)]e^{\mathrm i(\varphi_r+\varphi_s-\varphi_{r+s})}$ have positive coefficients. At $(k,M)=(16,64)$ the relation $4e_{16}$ completes $L_{64}(h)$, with $K_{4e_{16}}=a^4-6a^2(16d)^2+3(16d)^4$. If $K_{4e_{16}}\neq0$, the key is identifiable from these relations up to the cyclic relabeling.
\end{corollary}
Modulo the chain relations, $e_m\equiv me_1$, which proves generation; $\langle 4e_{16},h\rangle=64$ adds the missing coset. The coefficient of $4e_{16}$ vanishes at isolated values $\beta\approx0.063$ and $\beta\approx0.399$, where the third-order relations still recover the relative phases but completing the common phase would need another modular relation or a higher-order moment; it is nonzero at the evaluated $\beta\in\{0,0.3\}$. At $\beta=0$ all coefficients reduce to powers of $a>0$: the leakage is a property of the coupled harmonic means and does not require AN.

\begin{corollary}[odd harmonics]\label{cor:odd}
For $h=(1,3,\ldots,31)$ at $M=64$, no nonzero $c\in L_{64}(h)$ has $\|c\|_1\le3$, and fourth-order relations generate $L_{64}(h)$.
\end{corollary}
Signed sums of two harmonics lie strictly between $-64$ and $64$, and a signed sum of three odd numbers is odd. The relations $3e_1-e_2$ and $e_j-2e_{j+1}+e_{j+2}$ give $e_j\equiv h_je_1$, and the modular relation $w=2e_1+2e_{16}$ gives $64e_1\equiv0$, so the quotient is cyclic of order $64$ and the lattice is complete. Its coefficient $K_w=a^4-1086a^2d^2+2883d^4$ is negative at $\beta=0.3$, which the known sign correction absorbs. This odd set is also biorthogonal at $(k,M)=(16,64)$, so it differs from the consecutive family in geometry, relation multiplicity, coefficient magnitudes, and dependence among moments as well as in minimum relation order.

\begin{table}[!t]
\centering
\caption{Scope of the relation argument.}
\label{tab:scope}
\footnotesize
\begin{tabular}{@{}>{\raggedright\arraybackslash}p{0.34\linewidth}>{\raggedright\arraybackslash}p{0.56\linewidth}@{}}
\toprule
Keyed design & Exposing statistic and status\\
\midrule
Tones sharing one data parameter, phases $h_j\theta_s+\varphi_j$ & Relation moments~\eqref{eq:char}; key identified up to relabeling when $R=L_M(h)$ (this paper)\\
One common constellation rotation & Power-law carrier recovery~\cite{viterbi1983}; rotation identified up to the constellation symmetry\\
Independent symbols per subcarrier, per-subcarrier key & No pointwise cross-tone cancellation; other statistics, pilots, or code-aided attacks not analyzed here\\
Symbolwise mapping from a cryptographic generator & Security inherited from the generator, independent of the modulation geometry\\
\bottomrule
\end{tabular}
\end{table}

\section{Relation-Moment Attack}\label{sec:attack}

Eve uses a fixed relation set $\mathcal S=\mathcal S_0\cup\mathcal S_1$ of exact-zero relations $\mathcal S_0$ and modular relations $\mathcal S_1$ with $\langle c,h\rangle=M$. For the consecutive family, $\mathcal S_0=\{e_r+e_s-e_{r+s}:1\le r\le s,\ r+s\le16\}$ contains $64$ third-order relations and $\mathcal S_1=\{4e_{16}\}$ is the fourth-order completion. For the odd family, $\mathcal S_0$ contains the $455$ relations with $\|c\|_1=4$ and $\langle c,h\rangle=0$, one of each pair $\pm c$, and $\mathcal S_1$ the $297$ relations with $\|c\|_1=4$ and $\langle c,h\rangle=64$. At the evaluated $\beta\in\{0,0.3\}$ every $K_c$ in $\mathcal S$ is nonzero, and a Smith normal form computation confirms that $\mathcal S$ generates $L_{64}(h)$, of index $64$ in $\Z^{16}$, for both families, so by Proposition~\ref{prop:ann} the only residual ambiguity is the cyclic relabeling. Because $K_{4e_{16}}$ and $K_w$ vanish at isolated $\beta$, this nondegeneracy holds on the tested grid and not at every $\beta$.

\begin{algorithm}[!t]
\caption{Relation-moment attack on a segment $\mathcal I$, $|\mathcal I|=L$}
\label{alg:attack}
\begin{algorithmic}[1]
\small
\State $\widehat M_c\gets L^{-1}\sum_{\ell\in\mathcal I}U_{c,\ell}$, $\widehat\chi_c\gets\mathrm{sign}(K_c)\,\widehat M_c/|\widehat M_c|$, $c\in\mathcal S$
\State $w_c\gets|K_c|/\prod_jP_j^{|c_j|/2}$ with $P_j=a^2+h_j^2d^2+2\sigma_c^2$
\State initialize $\bm\psi$, $\psi_1=0$, by forward substitution over the generators of $\mathcal S_0$
\State refine $\bm\psi$ by Gauss--Newton on $\sum_{c\in\mathcal S_0}w_c^2r_c^2$, $r_c=\mathrm{Arg}(\widehat\chi_ce^{-\mathrm i\langle c,\bm\psi\rangle})$; stop after $40$ steps or when no coordinate moves by more than $10^{-12}$
\State $\widehat\varphi_1\gets M^{-1}\mathrm{Arg}\sum_{c\in\mathcal S_1}w_c\,\widehat\chi_ce^{-\mathrm i\langle c,\bm\psi\rangle}$
\State $\widehat\varphi_j\gets\psi_j+h_j\widehat\varphi_1\bmod2\pi$; map $\widehat{\bm\varphi}$ into $\mathcal F$
\State list $\widehat{\bm\varphi}-2\pi qh/M$ for $q\in\{-1,0,+1\}$; on $\mathcal K_b$, round each coordinatewise and add the other neighbour of the $3$ coordinates closest to a rounding boundary
\State compute the segment's matched symbol metrics under every candidate
\end{algorithmic}
\end{algorithm}

Algorithm~\ref{alg:attack} lists the attack, with $\mathrm{Arg}\in(-\pi,\pi]$. The weights depend only on public parameters, and no empirical moment is discarded for small magnitude. Since $\langle c,\bm\varphi\rangle=\langle c,\bm\psi\rangle+M\varphi_1$ for $c\in\mathcal S_1$, each modular relation estimates $e^{\mathrm iM\varphi_1}$, and step~5 averages them on the circle. The canonical representative is discontinuous at the edge of $\mathcal F$: a small error in $\varphi_1$ near that edge moves the estimate to the neighbouring cyclic representative and shifts every label of the segment by one, which the candidates $q=\pm1$ repair. On the grid the first coordinate is clipped to its protocol range, and at most $24$ candidates remain after duplicates are removed. Eve resolves the candidates of all segments of a codeword jointly with the code, in the spirit of code-aided synchronization~\cite{herzet2007}: starting from equal weights, she demaps each segment with the weighted mixture of its candidates' matched symbol likelihoods and decodes; if the decoder ends without a valid codeword, she reweights each segment's candidates by their likelihood under the decoder's extrinsic bit LLRs and decodes again, at most six times per codeword. A segment-local reference rule without decoding keeps the candidate with the fewest unsatisfied parity checks when its hard decisions fill the segment and all other coded bits are zero. The attack never enumerates the key space; per segment it costs the moment accumulation, one small fit, and one demapping per candidate, plus the decodes.

\section{Coded Experiments}\label{sec:results}

\subsection{Setup}

We take $(k,M)=(16,64)$, a regular $(3,6)$ LDPC code of length $1008$ and rate $506/1008$ with $50$ scaled min-sum iterations (factor $0.8$), and Gray labeling without a bit interleaver, as in~\cite{han2026arxiv}, so a codeword spans $n_s=168$ symbols and carries $R_c\log_2M\approx3.01$ information bits per symbol. We evaluate both harmonic families, $\beta\in\{0,0.3\}$, noise levels $\sigma_c\in\{0.10,0.13,0.16\}$, that is $4.95$, $2.67$ and $0.87$\,dB per complex slot, and the reuse lengths $L\in\{4,7,8,12,14,21,24,28,42,56,84,168\}$, all divisors of $168$. Every cell uses $N=500$ independent instances, each drawing a fresh message, fresh segment keys, AN, and noise. We simulate one receive realization, which Bob and Eve both process. Since only Eve's marginal error rate is reported, it equals that of an Eve with independent receiver noise at the same SNR. Intervals for one cell are exact Clopper--Pearson intervals~\cite{clopper1934}, and differences between independent cells use Newcombe's hybrid score interval~\cite{newcombe1998}. Bob decoded every codeword without error in every cell, so the one-sided $95\%$ upper bound of his block error rate is $1-0.05^{1/500}\approx6.0\cdot10^{-3}$ per cell. The simulation code and the enumerated relation sets will be released upon publication.

\subsection{Reuse dependence}

\begin{figure}[!t]
\centering
\includegraphics[width=.86\linewidth]{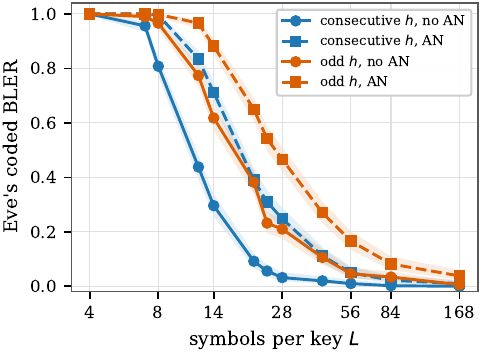}
\caption{Eve's coded $\BLER_E$ versus the number of symbols per key at $\sigma_c=0.13$, with $95\%$ Clopper--Pearson bands; $N=500$ instances per point.}
\label{fig:bler}
\end{figure}

Fig.~\ref{fig:bler} shows the central result. With one key per codeword, $L=168$, Eve's block error rate stays below $0.04$ at $\sigma_c=0.13$, and the attack meets the predeclared criterion, an upper end of the two-sided $95\%$ Clopper--Pearson interval of $\BLER_E$ at most $0.1$, in eleven of the twelve cells across all noise levels. The exception is the odd family with AN at $\sigma_c=0.16$, where $\BLER_E=0.096$ with upper end $0.125$. Shorter reuse raises Eve's error rate, but she stays successful well below one key per codeword: at $\sigma_c=0.13$ her error rate falls below $0.9$ between $L=8$ and $14$, crosses $0.5$ between $L=12$ and $28$, and drops below $0.05$ at $L=28$ for consecutive harmonics without AN and at $L=56$ to $168$ otherwise. Table~\ref{tab:noise} shows the same behavior at all three noise levels. On the same key estimates, the segment-local rule picks a wrong cyclic candidate in $36\%$ to $63\%$ of the segments at $\sigma_c=0.13$ and $L\in\{14,28,56\}$, and its error rate falls to $0.5$ only at $2.3$ to $6.3$ times the reuse length of the code-aided rule, so a weak label resolution overstates the resistance.

\begin{table}[!t]
\centering
\caption{Eve's coded $\BLER_E$ at three noise levels ($N=500$).}
\label{tab:noise}
\footnotesize
\setlength{\tabcolsep}{4pt}
\begin{tabular}{@{}llcccc@{}}
\toprule
$\sigma_c$ & $L$ & consec. & consec.+AN & odd & odd+AN\\
\midrule
$0.10$ & $14$ & $0.012$ & $0.206$ & $0.156$ & $0.482$\\
 & $28$ & $0.000$ & $0.030$ & $0.014$ & $0.098$\\
 & $168$ & $0.000$ & $0.006$ & $0.000$ & $0.006$\\
\midrule
$0.13$ & $14$ & $0.296$ & $0.712$ & $0.618$ & $0.882$\\
 & $28$ & $0.032$ & $0.248$ & $0.210$ & $0.466$\\
 & $168$ & $0.000$ & $0.010$ & $0.006$ & $0.038$\\
\midrule
$0.16$ & $14$ & $0.742$ & $0.958$ & $0.934$ & $0.992$\\
 & $28$ & $0.206$ & $0.598$ & $0.518$ & $0.794$\\
 & $168$ & $0.014$ & $0.026$ & $0.034$ & $0.096$\\
\bottomrule
\end{tabular}
\end{table}

\subsection{Artificial noise and relation order}

\begin{table}[!t]
\centering
\caption{Change of $\BLER_E$ at $\sigma_c=0.13$ with $95\%$ Newcombe intervals.}
\label{tab:effects}
\footnotesize
\setlength{\tabcolsep}{3pt}
\begin{tabular}{@{}lccc@{}}
\toprule
Comparison & $L=14$ & $L=28$ & $L=56$\\
\midrule
AN on$-$off, consec. & $.42\,[.36,.47]$ & $.22\,[.18,.26]$ & $.04\,[.02,.06]$\\
AN on$-$off, odd & $.26\,[.21,.31]$ & $.26\,[.20,.31]$ & $.12\,[.08,.16]$\\
odd$-$consec., no AN & $.32\,[.26,.38]$ & $.18\,[.14,.22]$ & $.04\,[.02,.06]$\\
odd$-$consec., AN & $.17\,[.12,.22]$ & $.22\,[.16,.27]$ & $.12\,[.08,.15]$\\
\bottomrule
\end{tabular}
\end{table}

As Table~\ref{tab:effects} shows, tangent AN raises Eve's measured error rate at intermediate reuse lengths, by up to $0.42$ at $\sigma_c=0.13$; for both families the $95\%$ interval of the difference excludes zero at $L=14$ and $28$ at all three noise levels. The observed delay is consistent with a lower signal-to-noise ratio of the empirical relation phases, since AN changes both the coefficients $K_c$ and the variances of the monomials. At $L=168$ Eve succeeds with or without AN. The odd family raises Eve's error rate by up to $0.34$ at $\sigma_c=0.13$, and the interval of the difference excludes zero in every cell in which either error rate lies between $0.01$ and $0.99$. Because the two families differ in relation order, relation multiplicity, coefficient magnitudes, moment dependence, and geometry at once, this difference may reflect their joint effect. Minimum relation order thus identifies which characters leak, while the rate at which they can be estimated and exploited depends on the full relation structure. Neither AN nor the odd family removes the one-codeword vulnerability.

The key-estimation step can be examined without decoding. Let $B_0$ and $B_1$ stack the relations of $\mathcal S_0$ and $\mathcal S_1$ as rows, $W_0=\mathrm{diag}(w_c:c\in\mathcal S_0)$, $\bm\omega$ the weights of $\mathcal S_1$ normalized to unit sum, and $G\in\R^{k\times(k-1)}$ the map that inserts $(\psi_2,\ldots,\psi_k)$ into $\bm\psi$ with $\psi_1=0$. Linearizing steps~4 to~6 of Algorithm~\ref{alg:attack} at the true key maps the relation-phase errors $\bm\varepsilon=(\bm\varepsilon_0,\bm\varepsilon_1)$ to the key error $\delta\bm\varphi=J\bm\varepsilon$ with
\begin{align}
J&=\bigl[\,P_0-M^{-1}h\,\bm\omega^{\mathsf T}B_1P_0\;\big|\;M^{-1}h\,\bm\omega^{\mathsf T}\,\bigr],\label{eq:jac}\\
P_0&=G\bigl((B_0G)^{\mathsf T}W_0^2B_0G\bigr)^{-1}(B_0G)^{\mathsf T}W_0^2.\nonumber
\end{align}
The per-symbol relation-phase vector $X_\ell$ has the entries $\mathrm{Im}(e^{-\mathrm i\langle c,\bm\varphi\rangle}\mathrm{sign}(K_c)U_{c,\ell})/|K_c|$, $c\in\mathcal S$, whose joint law does not depend on the key. We estimate $\Sigma=\mathrm{Cov}(X_\ell)$, cross-covariances included, from $2\cdot10^5$ symbols per configuration in a Monte Carlo run independent of the coded experiments. The Gaussian plug-in approximation $\delta\bm\varphi\sim\mathcal N(0,J\Sigma J^{\mathsf T}/L)$, scored by the RMS error minimized over the cyclic relabeling, gives the predicted median key error in Table~\ref{tab:keyerr}. It applies to the continuous estimate before canonicalization and candidate selection, in the local regime where all retained moments are nonzero and the residuals stay on one wrapped branch. Without AN, and for the odd family with AN, it agrees with the median of the same error measured in the coded runs within $12\%$ for $L\ge21$ at all three noise levels. For the consecutive family with AN, whose common phase rests on the single high-kurtosis fourth-order moment of $4e_{16}$, the local prediction overstates the error by up to $44\%$. AN raises both the predicted and the measured key error at every tested $L$, consistent with a loss of relation-phase SNR. The odd key is estimated more accurately than the consecutive key at every tested $L$ without AN and for $L\ge14$ with AN, and at a median key error near $0.11$\,rad, $\BLER_E$ is $0.21$ for the odd family at $L=28$ and $0.02$ for the consecutive family at $L=42$. At matched median key error, Eve's decoding is thus more sensitive in the odd family, consistent with its different constellation geometry.

\begin{table}[!t]
\centering
\caption{Median RMS key error (rad) of the continuous estimate, minimized over the cyclic relabeling, at $\sigma_c=0.13$: plug-in prediction from~\eqref{eq:jac} / coded runs.}
\label{tab:keyerr}
\footnotesize
\setlength{\tabcolsep}{4pt}
\begin{tabular}{@{}lccc@{}}
\toprule
Configuration & $L=28$ & $L=56$ & $L=168$\\
\midrule
consec. & $.139/.135$ & $.098/.098$ & $.057/.054$\\
consec.+AN & $.261/.204$ & $.214/.164$ & $.127/.116$\\
odd & $.105/.107$ & $.074/.074$ & $.043/.044$\\
odd+AN & $.159/.170$ & $.113/.116$ & $.065/.065$\\
\bottomrule
\end{tabular}
\end{table}

\subsection{Finite key protocol and key-refresh cost}

\begin{figure}[!t]
\centering
\includegraphics[width=.95\linewidth]{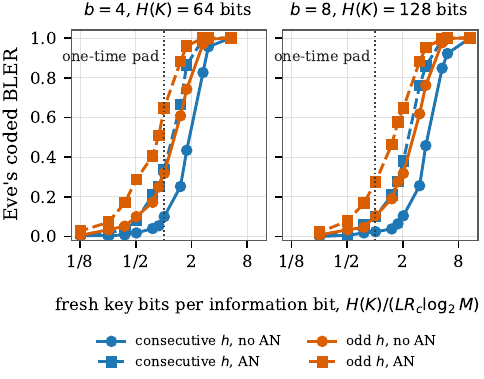}
\caption{Eve's coded $\BLER_E$ on the key grid $\mathcal K_b$ versus fresh key bits per information bit for the twelve tested $L$ at $\sigma_c=0.13$; $N=500$.}
\label{fig:cost}
\end{figure}

To tie the attack to a key entropy, we repeat the experiment with keys drawn from $\mathcal K_4$ and $\mathcal K_8$ at $\sigma_c=0.13$ and all twelve reuse lengths, where Eve projects her estimate onto the grid as described in Sec.~\ref{sec:attack}. The grid and continuous-key error rates differ by at most $0.064$ over the $96$ cells, and only $4$ of the $96$ Newcombe intervals of the difference exclude zero, each with the lower rate on the grid. The selected candidate equals the transmitted grid key, which is exact recovery modulo the cyclic relabeling because both are canonical, in less than $10\%$ of segments. A key refreshed every $L$ symbols consumes $H(K)/L$ fresh secret bits per symbol, or $\rho_K=H(K)/(LR_c\log_2M)$ fresh key bits per information bit (Fig.~\ref{fig:cost}).

To summarize the curves, we call a reuse length high-error for a configuration when the pointwise one-sided $95\%$ Clopper--Pearson lower bound of $\BLER_E$ is at least $\eta_E$, and Table~\ref{tab:cost} reports $\rho_K$ at the largest such tested $L$. We report the target $\eta_E=0.9$, chosen after the runs, and $\eta_E=0.5$ as a sensitivity check. On the $2^{128}$-key grid at $\eta_E=0.9$, the cheapest tested schedule that passes this rule uses $3.54$ fresh key bits per information bit, for the odd family with AN, and the consecutive family without AN needs $10.62$ because its lower bound at $L=7$ is $0.899$. Bonferroni-simultaneous bounds over the $96$ cells leave Table~\ref{tab:cost} unchanged. No entry is below $1.01$, so every tested schedule that meets either target on either grid consumes at least as much fresh key as a one-time pad on the information bits, which provides perfect secrecy at one key bit per information bit~\cite{shannon1949}.

\begin{table}[!t]
\centering
\caption{$\rho_K$ at the largest tested $L$ whose pointwise one-sided $95\%$ lower bound of $\BLER_E$ is at least the post-hoc target $\eta_E$ ($\sigma_c=0.13$, $N=500$).}
\label{tab:cost}
\footnotesize
\setlength{\tabcolsep}{4pt}
\begin{tabular}{@{}llcccc@{}}
\toprule
$\eta_E$ & $H(K)$ & consec. & consec.+AN & odd & odd+AN\\
\midrule
$0.9$ & $128$ & $10.62$ & $5.31$ & $5.31$ & $3.54$\\
 & $64$ & $3.04$ & $2.66$ & $2.66$ & $1.77$\\
\midrule
$0.5$ & $128$ & $5.31$ & $3.04$ & $3.04$ & $1.77$\\
 & $64$ & $2.66$ & $1.52$ & $1.52$ & $1.01$\\
\bottomrule
\end{tabular}
\end{table}

\section{Discussion}\label{sec:disc}

The resource implications depend on how the segment keys are supplied. Under independent random refresh on the $2^{128}$-key grid, the cheapest tested schedule that passes the post-hoc $0.9$ criterion uses $3.54$ times the key of a one-time pad on the information bits, and no tested schedule on either grid meets even the $0.5$ criterion with less key than the pad. A cryptographic generator that expands a short master key makes frequent refresh cheap, and confidentiality then rests on that generator, which could equally encrypt the bits. Physical-layer key generation from reciprocal channels~\cite{maurer1993,zhang2016keygen} has to be assessed by its net secret-key rate after probing, reconciliation, and privacy amplification. In all three cases the modulation geometry itself creates no fresh secret entropy and no computational security guarantee. Refresh acts on the relation leakage; the first-tone channel of Sec.~\ref{sec:model} depends on the labeling convention and closes with a full-circle first coordinate.

Stronger attacks can only help Eve: refining the key estimates themselves with the decoder output, known frame structure such as pilots, headers, and CRCs, or multiple antennas would lower her error rates further. Retransmission protocols also lengthen the reuse: Chase HARQ repeats the same keyed symbols in every round~\cite{han2026chase}, and since the relation monomials do not depend on the symbols, each round adds $n_s$ independent observations under the same key unless the key is refreshed per round. The operational transition is measured for one regular $(3,6)$ code, labeling, and decoder; the algebraic leakage is independent of these choices, whereas the reuse length at which it turns into message recovery can depend on them. Channel estimation errors and front-end impairments are left for future work.

\section{Conclusion}\label{sec:concl}

Reusing a phase key in harmonically coupled modulation converts short modular relations among the harmonic indices into estimable key characters. In the coded Fourier-curve case study, the evaluated relation-moment attack recovers the message reliably when one key spans a 168-symbol codeword and, for consecutive harmonics without AN, when it spans only $28$ symbols; shorter reuse raises its block error rate. Tangent AN and the tested harmonic family without relations below order four increase Eve's measured block error rate at several intermediate reuse lengths and do not remove the one-codeword vulnerability. For the tested product-grid protocol, independent refresh frequent enough to keep Eve's block error rate high consumes at least as much secret key as a one-time pad on the information bits under both tested targets. Nominal key-space size and the error rate of key-less receivers are therefore insufficient security evidence for keyed modulations with repeated waveform structure.

\end{document}